\documentclass[10pt]{article}
\usepackage{}

\usepackage{pgf}
\usepackage[square, comma, sort&compress, numbers]{natbib}
\usepackage[english]{babel}
\usepackage[utf8]{inputenc}
\usepackage{geometry}
\usepackage{array}
\usepackage{xspace}

\usepackage{amsmath,amsfonts,amssymb,amsopn,amsthm,amscd}
\theoremstyle{plain}
\usepackage{graphicx} 
\usepackage{epstopdf}
\usepackage{enumerate}

\def\Box{\vcenter{\vbox{\hrule\hbox{\vrule
     \vbox to 8.8pt{\hbox to 10pt{}\vfill}\vrule}\hrule}}}
\usepackage{longtable,array}
\usepackage{multirow,multicol}
\newcommand{\tabincell}[2]{\begin{tabular}{@{}#1@{}}#2\end{tabular}}
\usepackage{enumerate}

\usepackage{indentfirst}
\newcommand{\Zz}{{\mathbb Z}}

\newcommand\cc{{\mathcal C}}        %

\newcommand\cF{{\mathbf c}}

\newtheorem{thm}{Theorem}[section]
\newtheorem{lem}[thm]{Lemma}

\date{}

\begin{document}

%

\title{Asymptotically optimal codebooks derived from generalised bent functions\thanks{The work was supported by   the Science and
Technology Development Fund of Tianjin Education Commission
for Higher Education No. 2018KJ215.}
}

\author{Qiuyan Wang\thanks{Q. Wang is with the School of Computer Science and
Technology, Tiangong University, Tianjin,
300387, China, and with the Provincial Key Laboratory of Applied Mathematics, Putian University, Putian, Fujian 351100, China. Email: wangyan198801@163.com},
Yang Yan\thanks{Y. Yan is with the School of Information Technology and Engineering, Tianjin University of Technology and Education, Tianjin, 300222, Yang Yan is the corresponding author, Email: yanyangucas@126.com},
Chenhuang Wu\thanks{C. Wu is with the Provincial Key Laboratory of Applied Mathematics, Putian University. Email:  ptuwch@163.com} and
Chun'e Zhao\thanks{College of Science,
China University of Petroleum, Qingdao 266555, Shandong, China, Email: zhaochune1981@163.com}
}
\maketitle

\let\thefootnote\relax\footnotetext{}

\let\thefootnote\relax\footnotetext{}

\begin{abstract}
Codebooks are required to have small inner-product correlation in many practical applications, such as direct spread code division multiple access
communications, space-time codes and compressed sensing. In general, it is difficult to construct optimal codebooks. In this paper, two kinds of codebooks
are presented and  proved to optimally optimal with respect to the welch bound. Additionally, the constructed codebooks in this paper have new parameters.

{\bf Keywods}: codebook, asymptotic optimality, Welch bound, generalised bent function.

\end{abstract}

\section{Introduction}
An $(N,K)$ codebook $\cc$ is defined to be a set of unit-norm complex vectors $\{\mathbf{c}_i\}_{i=0}^{N-1} $ in $\mathbb{C}^K$  over an alphabet \textit{A}. Let
\begin{align*}
I_{\max}(\cc)=\underset{0\leq i\neq j \leq N-1}{\max}|\mathbf{c}_i\mathbf{c}_j^H|,
\end{align*}
where $\mathbf{c}_j^H$ denotes the Hermite transpose of vector $\mathbf{c}_j$. The maximum inner-product correction $I_{\max}(\cc)$ is a performance measure of a codebook $\cc$ in practical applications. In code division multiple access (CDMA) systems, one important problem is to minimize the codebook’s maximal cross-correlation amplitude $I_{\max}(\cc)$.

For a given $K$, it is usually  desirable that $N$ is as large as possible and $I_{\max}(\cc)$ is as small as possible simultaneously. However, the parameters $N$, $K$ and
$I_{\max}(\cc)$ of a codebook have to satisfy the Welch bound \cite{Welch}. That is, there is a tradeoff between the codeword length $K$, the set size $N$ and $I_{\max}(\cc)$.
A codebook meeting the theoretical bound with equality is said to be optimal. Searching optimal codebooks has been an interesting research topic in recent years. Many classes of optimal codebooks has been constructed \cite{Apl.cd3,MWEB,Apl.pa,MWEB1,Apl.cd1,Apl.cd2,F1,F2,F3,F4}.

It is worthwhile to point out that the constructed codebooks so far have restrictive parameters $N$ and $K$. Hence, many researchers attempt to research asymptotically optimal
codebooks, i.e., $I_{\max}(\cc)$ asymptotically meets the theoretical bound for sufficiently large $K$. One important method to construct asymptotically optimal codebooks is from difference sets of finite abelian gruops which is developed by Ding and Feng \cite{Apl.cd1,Apl.cd2}. In \cite{Apl.cd1,Apl.cd2}, several series of asymptotically codebooks were constructed by using almost difference  sets. In \cite{NM1,NM3Yu,Cao}, asymptotically optimal codebooks were presented by binary row selection sequences. Character sums over finite fields are considered to be useful tools for the design of asymptotically codebooks. Recently, in \cite{AL2,Heng}, Heng et al. obtained two new constructions of infinitely many codebooks by Jacobi sums and their generalizations. In \cite{Luo1,Luo2}, Luo and Cao defined a new character sum called hyper Eisenstein sum and presented two constructions of infinitely many new codebooks achieving the Wech bound. In \cite{Tian}, Tian presented two constructions of codebooks with additive characters over finite fields.

Bent functions are a class of Boolean functions and have important applications in cryptography, code theory and sequences for communications. In cryptography, bent vectorial functions can be used as substitution boxes in block ciphers (ensuring confusion, as explained by Shannon \cite{Sha}. In code theory, they are useful for constructing error correcting codes (Kerdock codes) \cite{Mac}. In sequences, they permit to construct sequences with low correlation \cite{Kar}.

The objective of this paper is to present two constructions of codebooks using generalised bent functions over a ring of integers modulo $Q$. The presented two kinds of codebooks have properties: (1) they are asymptotically optimal with respect to the Welch bound; (2) the parameters of these codebooks are new and flexible. As a comparison with the known ones,  our codebooks are listed in Table \ref{CCCtabl1}.
\begin{table}[!htbp]
 \caption{\small The parameters of codebooks asymptotically meeting the Welch bound}
 \label{CCCtabl1}
 \centering
 \begin{tabular}{|l|l|l|l|l|}
   \hline
   References& Parameters $(N,K)$  &Constraints  \\ \hline
   \cite{NM1}& $\left(p^n,\frac{p-1}{2p}\left(p^n+p^{n/2}\right)+1\right)$   & $p$ is an odd prime  \\ \hline
   \cite{NM5} & $\left(q^2,\frac{(q-1)^2}{2}\right)$ &  $q$ is an odd prime power  \\ \hline
   \cite{LiCC} & $\left(q(q+4),\frac{(q+3)(q+1)}{2}\right)$ &  $q$ is a power of a prime\\ \hline
    \cite{LiCC} & $\left(q,\frac{q+1}{2}\right)$ &  $q$ is a prime power  \\ \hline
   \cite{NM3Yu} & $\left(p^n-1,\frac{p^n-1}{2}\right)$  & $p$ is an odd prime  \\ \hline
   \cite{zhou} & $\left(q^t+q^{t-1}-1,q^{t-1}\right)$   & \tabincell{c}{$q$ is a prime power}  \\ \hline
   \cite{AL2} & $\left((q-1)^\ell+q^{\ell-1},q^{\ell-1}\right)$  & \tabincell{c}{$q\geq4$ is a power of a prime, \\ $\ell>2$}  \\ \hline
    \cite{AL2} & $\left((q-1)^\ell+M,M \right)$  & \tabincell{c}{$M=\frac{(q-1)^\ell+(-1)^{\ell+1}}{q}$, \\ $q$ is a prime power and $\ell>2$ } \\ \hline
   \cite{Luo} & $\left((q^s-1)^m+M,M\right)$ &\tabincell{c}{ $M=\frac{(q^s-1)^m+(-1)^{m+1}}{q}$, \\$s>1$, $m>1$, \\ $q$ is a prime power}  \\ \hline
   \cite{Luo} & $\left((q^s-1)^m+q^{sm-1},q^{sm-1}\right)$  & \tabincell{c}{ $s>1$, $m>1$, \\ $q$ is a prime power}   \\ \hline
   \cite{Luo} & $\left((q^s-1)^m+q^{sm-1},q^{sm-1}\right)$  & \tabincell{c}{ $s>1$, $m>1$, \\ $q$ is a prime power}   \\ \hline
   \cite{Tian} & $\left(q^3+q^2,q^2\right)$  & $q$ is a prime power   \\ \hline
   \cite{Tian} & $\left(q^3+q^2-q,q^2-q\right)$  & $q$ is a prime power  \\ \hline
    Theorem \ref{thm1} & $\left((p_{{\min}}+1)Q^2,Q^2\right)$  & \tabincell{c}{ $Q>1$ is an integer and \\ $p_{{\min}}$ is the smallest \\ prime factor of $Q$}  \\ \hline
     Theorem \ref{thm3} & $\left((p_{{\min}}+1)Q^2-Q,Q(Q-1)\right)$  & \tabincell{c}{ $Q>2$ is an integer and \\ $p_{{\min}}$ is the smallest \\ prime factor of $Q$}  \\ \hline
 \end{tabular}
 \end{table}

This paper is organized as follows. In Section 2, we briefly recall some definitions and notations which will be needed in our discussion. In Section 3 and Section 4, we present our constructions of  codebooks and prove they are asymptotically optimal with respect to the Welch bound. In Section 5, concluding remarks of this paper is given.

\section{Preliminaries}

In this section, we present some notations and preliminaries  which are needed for the proof of the main results. Firstly, a useful lemma is given in the following.
 \begin{lem}[Linear Congruence Theorem, \cite{SJ}]\label{lem1}
 Let $a$, $c$, and $m$ be integers with
$m>1$, and let $g=\gcd(a,m)$.
\begin{enumerate}[(1)]
\item If $g\nmid c$, then the congruence $ax\equiv c \pmod {m}$ has no solutions.
\item If $g\mid c$, then the congruence $ax\equiv c \pmod {m}$ has exactly $g$ incongruent solutions.
\end{enumerate}
\end{lem}

The following is the well-known Welch bound on $N$, $K$ and $I_{\max}(\cc)$ of a codebook $\cc$.
\begin{lem}{\rm \cite{Welch}}\label{welch}
For any $(N,K)$ codebook $\cc$ with $N>K$,
$$
I_{\max}(\cc)\geq I_{W}=\sqrt{\frac{N-K}{(N-1)K}}.
$$
Moreover, the equality holds if and only if for all pairs of $(i,j)$ with $i\neq j$, it holds that
$$
|\mathbf{c}_i\mathbf{c}_j^H|=\sqrt{\frac{N-K}{(N-1)K}}.
$$
\end{lem}

Next, we introduce the definition of generalised bent functions. Let $Q$ be a positive integer and $\Zz_Q$ the ring of integers modulo $Q$. Assume that $\xi_Q$ is a primitive $Q$-th root of unity. Denote by $\Zz_Q^m$ the $m$-dimensional vector space over $\Zz_Q$. A function mapping from $\Zz_Q^m$ to $\Zz_Q$ is termed a generalised Boolean function on $m$ variables. For a generalised Boolean function $f$, if the complex Fourier coefficients
$$
F_f(\textbf{a})=\frac{1}{\sqrt{Q^m}}\sum_{\textbf{x}\in \Zz_Q^m}\xi_Q^{f(\textbf{x})-\textbf{a}^T\cdot\textbf{x}}
$$
preserve unit magnitude for any $\textbf{a}\in\Zz_Q^m$, then $f$ is a generalised bent function.

Bent functions are a hot research topic due to their wide applications in cryptography, information theory and coding theory. Kumar et al. \cite{Kumar} introduced the definition of generalised bent functions from $\Zz_Q^m$ to $\Zz_Q$ and gave a class of generalised bent functions in the following lemma.

\begin{lem}[\cite{Kumar}]\label{lem2}
Assume that $Q$ is a positive integer. Let $\Omega(x)$ be an arbitrary permutation and $\Theta(x)$ an arbitrary function on $\Zz_Q$. Then the function
$$
f(x_1,x_2)=x_2\Omega(x_1)+\Theta(x_1)
$$
is generalised bent, where $x_1,x_2\in\Zz_Q$.
\end{lem}

For more details on bent functions and generalised bent functions, we refer readers to \cite{Bent1}. Inspired by the generalised bent functions given in Lemma \ref{lem2}, we propose two constructions of codebooks in the following two sections.

\section{The first construction of asymptotically optimal codebooks}\label{S3}

In this section, we present a construction of codebooks and show that the maximum inner-product correction of these codebooks asymptotically achieves the Welch bound. Before proposing our construction, we need to do some preparations.

Suppose that $Q>1$ is an integer and $p_{{\rm min}}$ is the smallest prime factor of $Q$. Denote the standard basis of the $Q^2$-dimensional Hilbert space by $\mathcal{E}_{Q^2}$ which is formed by $Q^2$ vectors of length $Q^2$ as follow:
\begin{eqnarray*}
&&(1,0,0,\cdots,0,0),\\
&&(0,1,0,\cdots,0,0),\\
&&\ \ \ \ \ \ \ \ \ \ \ \ \vdots\\
&&(0,0,0,\cdots,0,1).
\end{eqnarray*}

In the following theorem, we give a new construction of  infinite  many codebooks and evaluate their maximum inner-product correction.

\begin{thm}\label{thm1}
Let symbols be the same as above. For any $a\in\Zz_{p_{{\min}}}$,  $b,u\in\Zz_Q$, define a unit-norm complex vector of length $Q^2$ by
$$
\cF_{a,b,u}=\frac{1}{Q}\left(\xi_Q^{j\left(a\pi(i)+b\right)+u\sigma(i)}\right)_{i,j\in\Zz_Q},
$$
where $\pi(x)$ and $\sigma(x)$ are permutations on $\Zz_Q$. Let
\begin{align}
&\mathcal{F}=\left\{\cF_{a,b,u}:a\in\Zz_{p_{{\min}}},b,u\in\Zz_Q\right\},\nonumber\\
&\cc=\mathcal{F}\cup\mathcal{E}_{Q^2}. \label{def-1}
\end{align}
Then the set $\cc$ is a $\left((p_{{\min}}+1)Q^2,Q^2\right)$ codebook with $I_{\max}(\cc)=\frac{1}{Q}$.
\end{thm}
\begin{proof}
According to the definition of codebooks, we deduce that $\cc$ is consisted of $(p_{{\rm min}}+1)Q^2$ codewords with length $Q^2$. In other words, $\cc$ is a codebook with parameters $\left((p_{{\min}}+1)Q^2,Q^2\right)$. Now we turn to the computation of $I_{\max}(\cc)$. Let $\cF_1$, $\cF_2\in \cc$ be two distinct codewords. We distinguish among the following three cases to calculate the correlation of $\cF_1$  and $\cF_2$.

(1) If $\cF_1$, $\cF_2\in\mathcal{E}_{Q^2}$, it is easy to verify that $\left|\cF_1\cF_2^H\right|=0$.

(2) If $\cF_1\in\mathcal{E}_{Q^2}$ and $\cF_2\in\mathcal{F}$, it is obvious that $\left|\cF_1\cF_2^H\right|=\frac{1}{Q}$.

 (3) If $\cF_1$, $\cF_2\in\mathcal{F}$, write $\cF_1=\cF_{a,b,u}$ and $\cF_2=\cF_{\widehat{a},\widehat{b},\widehat{u}}$, where $(a-\widehat{a},b-\widehat{b},u-\widehat{u})\neq(0,0,0)$. Then we deduce that
\begin{eqnarray*}
\cF_1\cF_2^H&=&\frac{1}{Q^2}\sum_{i=0}^{Q-1}\sum_{j=0}^{Q-1}\xi_Q^{j\left(a\pi(i)+b\right)+u\sigma(i)-j\left(\widehat{a}\pi(i)+\widehat{b}\right)-\widehat{u}\sigma(i)}\\
&=&\frac{1}{Q^2}\sum_{i=0}^{Q-1}\xi_Q^{(u-\widehat{u})\sigma(i)}\sum_{j=0}^{Q-1}\xi_Q^{j\left((a-\widehat{a})\pi(i)+b-\widehat{b}\right)}.
\end{eqnarray*}

If $a-\widehat{a}=0$, we obtain
\begin{eqnarray*}
\cF_1\cF_2^H&=&\frac{1}{Q^2}\sum_{i=0}^{Q-1}\xi_Q^{(u-\widehat{u})\sigma(i)}\sum_{j=0}^{Q-1}\xi_Q^{j(b-\widehat{b})}\\
&=&0,
\end{eqnarray*}
where the second identity follows from the fact that $(b-\widehat{b},u-\widehat{u})\neq(0,0)$.

If $a-\widehat{a}\neq 0$, then $\gcd(a-\widehat{a},Q)=1$. By Lemma \ref{lem1}, the congruence $(a-\widehat{a})\pi(x)+b-\widehat{b}\equiv 0 \pmod {Q}$ has only one integer solution $i^{\prime}$ in $\mathbb{Z}_{Q}$. In this case, we have
\begin{eqnarray*}
\cF_1\cF_2^H&=&\frac{1}{Q^2}\sum_{i=0}^{Q-1}\xi_Q^{(u-\widehat{u})\sigma(i)}\sum_{j=0}^{Q-1}\xi_Q^{j\left((a-\widehat{a})\pi(i)+b-\widehat{b}\right)}\\
&=&\frac{1}{Q^2}\sum_{i=0,i\neq i^{\prime}}^{Q-1}\xi_Q^{(u-\widehat{u})\sigma(i)}\sum_{j=0}^{Q-1}\xi_Q^{j\left((a-\widehat{a})\pi(i)+b-\widehat{b}\right)}+\frac{1}{Q^2}\xi_Q^{(u-\widehat{u})\sigma(i^{\prime})}\sum_{j=0}^{Q-1}1\\
&=&\frac{1}{Q}\xi_Q^{(u-\widehat{u})\sigma(i^{\prime})},
\end{eqnarray*}
where the last equality is derived from the fact that $\sum_{j=0}^{Q-1}\xi_Q^{\ell j}=0$ for any $\ell\not\equiv 0 \pmod {Q}$. Therefore, we obtain $\left|\cF_1\cF_2^H\right|\in\{0,\frac{1}{Q}\}$.

The analysis above shows that $I_{\max}(\cc)=\frac{1}{Q}$. This completes the proof of this theorem.
\end{proof}

The next theorem deals with the asymptotical optimality of the codebooks defined in Theorem \ref{thm1}.

\begin{thm} \label{thm2}
Let  symbols be the same as before. Then the maximum inner-product correction $I_{\max}(\cc)$  of the codebook $\cc$ defined in \eqref{def-1} asymptotically meets the Welch bound.
\end{thm}
\begin{proof}
From Theorem \ref{thm1}, we know $\cc$ is a $\left((p_{{ \min}}+1)Q^2,Q^2\right)$ codebook. Obviously, $Q^2<(p_{{\rm min}}+1)Q^2<Q^4$
Note that the corresponding Welch bound of $\cc$ is
$$
I_W=\sqrt{\frac{p_{{\min}}}{p_{{\min}}Q^2+Q^2-1}}.
$$
Then we obtain
$$
\frac{I_{\max}(\cc)}{I_W}=\sqrt{\frac{p_{{\min}}Q^2+Q^2-1}{Q^2p_{{\min}}}}=\sqrt{1+\frac{1}{p_{{\min}}}-\frac{1}{Q^2p_{{\min}}}}.
$$
Observe that
$$\lim_{p_{{\min}}\rightarrow+\infty }\frac{I_{\max}(\cc)}{I_W}=1,$$
which implies that the codebook $\cc$ asymptotically meets the Welch bound.
\end{proof}

 In Table \ref{tab2}, we list some examples of codebooks generated by Theorem \ref{thm1}. The numerical results indicate that the codebooks defined in Theorem \ref{thm1} asymptotically achieve the Welch bound as $p_{\min}$ increases, as predicted in Theorem \ref{thm2}.
\begin{table}[!htbp]
 \caption{\small The parameters of the $(N,K)$ codebook in Theorem \ref{thm1}}
 \label{tab2}
 \centering
 \begin{tabular}{|l|l|l|l|l|l|l|}
   \hline
   $p_{\min}$& $Q$   &$N$ &$K$ &$I_{\max}$ & $I_W$& $I_W/I_{\max}$ \\ \hline
   5& $35$   & $7350$  &1225 & 0.02857& 0.02608 & 0.91293    \\ \hline
   13 & $221$ &  $683774$ & 48841&$0.45249\times 10^{-2}$ & $0.43603\times 10^{-2}$&  0.96362 \\ \hline
   17 & $493$ &  $4374882$  & 243049& $0.20284\times 10^{-2}$&$0.19712\times 10^{-2} $& 0.97183  \\ \hline
    31 & $1891$ &  $114428192$  &3575881 &$0.52882\times 10^{-3}$ & $0.52049\times 10^{-3}$& 0.98425\\ \hline
   43 & $3053$  & $410115596$  & 9320809& $0.32755\times 10^{-3}$&$0.32380\times 10^{-3}$ & 0.98857\\ \hline
  61 & $4453$   & 1229410598 & 19829209&$0.22345\times 10^{-3}$ &$0.22275\times 10^{-2}$ &0.99190 \\ \hline
   73 & $6497$  & 3123614666 & 42211009&$0.15392\times 10^{-3}$ &$0.15287\times 10^{-3}$ & 0.99322 \\ \hline
    83 & $7387$  & 4583692596 &54567796 &$0.13537\times 10^{-3}$ &$0.13456\times 10^{-3}$ &0.99403\\ \hline
   97 & $10961$ &11774065058 & 120143521&$0.91230\times 10^{-4}$ &$0.90766\times 10^{-4}$ &0.99488 \\ \hline
 \end{tabular}
 \end{table}

\section{The second construction of asymptotically optimal codebooks}

In this section, we propose a construction of codebooks by slightly modifying the construction of codebooks in Section \ref{S3}. In addition, we show that these codebooks are asymptotically optimal with respect  to the Welch bound.

Let $Q>1$ be an integer and $p_{{\min}}$ the smallest prime factor of $Q$. Assume that $\mathcal{E}_{Q(Q-1)}$ is a set consisted of all rows of the identity matrix $I_{Q(Q-1)}$. That is to say, $\mathcal{E}_{Q(Q-1)}$ is the standard basis of the Hilbert space with dimension $Q(Q-1)$. Let $\pi(x)$ and $\sigma(x)$ be permutations on $\Zz_Q$. For any $\ell\in\Zz_Q$, we define a set by
\begin{equation}\label{con1}
\mathcal{F}=\left\{\cF_{a,b,u}:a\in\Zz_{p_{{\rm min}}},b,u\in\Zz_Q\right\},
\end{equation}
where
$$
\cF_{a,b,u}=\frac{1}{\sqrt{Q(Q-1)}}\left(\xi_Q^{j(a\pi(i)+b)+u\sigma(i)}\right)_{i\in\Zz_Q\setminus\{\ell\},j\in\Zz_Q}.
$$
Let
\begin{align}\label{def-2}
\cc=\mathcal{F}\bigcup \mathcal{E}_{Q(Q-1)}.
\end{align}
Then $\cc$ is a codebook with parameters $\left(p_{{\min}}Q^2+Q^2-Q,Q(Q-1)\right)$ and the the maximum inner-product correction $I_{\max}(\cc)$ of $\cc$ can be obtained in the following theorem.

\begin{thm}\label{thm3}
Assume that $Q>1$ is an integer and $p_{{\rm min}}$ is the smallest prime factor of $Q$.  Then $\cc$  defined by \eqref{def-2} is a $\left(p_{{\min}}Q^2+Q^2-Q,Q(Q-1)\right)$ codebook with $I_{\max}(\cc)=\frac{1}{\sqrt{Q(Q-1)}}$.
\end{thm}
\begin{proof}
Let $\cF_1$, $\cF_2$ be two distinct codewords in $\cc$. Now we calculate the correlation of  $\cF_1$ and  $\cF_2$ by distinguishing among the following cases.

(1) If $\cF_1$, $\cF_2\in\mathcal{E}_{Q(Q-1)}$, it is obvious that $\left|\cF_1\cF_2^H\right|=0$.

(2) If $\cF_1\in\mathcal{E}_{Q(Q-1)}$ and $\cF_2\in\mathcal{F}$, we have  $\left|\cF_1\cF_2^H\right|=\frac{1}{\sqrt{Q(Q-1)}}$.

 (3) If $\cF_1$, $\cF_2\in\mathcal{F}$, we assume that  $\cF_1=\cF_{a,b,u}$ and $\cF_2=\cF_{\widehat{a},\widehat{b},\widehat{u}}$, where $(a-\widehat{a},b-\widehat{b},u-\widehat{u})\neq(0,0,0)$. Then we have
\begin{eqnarray*}
\cF_1\cF_2^H&=&\frac{1}{Q(Q-1)}\sum_{i=0,i\neq\ell}^{Q-1}\sum_{j=0}^{Q-1}\xi_Q^{j(a\pi(i)+b)+u\sigma(i)-j(\widehat{a}\pi(i)+\widehat{b})-\widehat{u}\sigma(i)}\\
&=&\frac{1}{Q(Q-1)}\sum_{i=0,i\neq\ell}^{Q-1}\xi_Q^{(u-\widehat{u})\sigma(i)}\sum_{j=0}^{Q-1}\xi_Q^{j((a-\widehat{a})\pi(i)+b-\widehat{b})}.
\end{eqnarray*}

If $a-\widehat{a}=0$ and $b-\widehat{b}\neq 0$, `we obtain
\begin{eqnarray*}
\cF_1\cF_2^H&=&\frac{1}{Q(Q-1)}\sum_{i=0,i\neq\ell}^{Q-1}\xi_Q^{(u-\widehat{u})\sigma(i)}\sum_{j=0}^{Q-1}\xi_Q^{j(b-\widehat{b})}\\
&=&0.
\end{eqnarray*}

If $a-\widehat{a}=0$ and $b-\widehat{b}=0$, it follows from the fact that $u-\widehat{u}\not\equiv 0 \pmod {Q}$ that
\begin{eqnarray*}
\cF_1\cF_2^H&=&\frac{1}{Q(Q-1)}\sum_{i=0,i\neq\ell}^{Q-1}\xi_Q^{(u-\widehat{u})\sigma(i)}\sum_{j=0}^{Q-1}1\\
&=&\frac{1}{Q-1}\sum_{i=0}^{Q-1}\xi_Q^{(u-\widehat{u})\sigma(i)}-\frac{1}{Q-1}\xi_Q^{(u-\widehat{u})\sigma(\ell)}\\
&=&-\frac{1}{Q-1}\xi_Q^{(u-\widehat{u})\sigma(\ell)}.
\end{eqnarray*}

If $a-\widehat{a}\neq 0$, then $\gcd(a-\widehat{a},Q)=1$. According to Lemma \ref{lem1}, the congruence $(a-\widehat{a})\pi(x)+b-\widehat{b}\equiv 0 \pmod {Q}$ has only one integer solution $i^{\prime}$ in $\mathbb{Z}_{Q}$.

When  $i^{\prime}\neq\ell$, we can  deduce that
\begin{eqnarray*}
\cF_1\cF_2^H&=&\frac{1}{Q^(Q-1)}\sum_{i=0,i\neq\ell}^{Q-1}\xi_Q^{(u-\widehat{u})\sigma(i)}\sum_{j=0}^{Q-1}\xi_Q^{j((a-\widehat{a})\pi(i)+b-\widehat{b})}\\
&=&\frac{1}{Q^(Q-1)}\sum_{i=0,i\neq i^{\prime},\ell}^{Q-1}\xi_Q^{(u-\widehat{u})\sigma(i)}\sum_{j=0}^{Q-1}\xi_Q^{j((a-\widehat{a})\pi(i)+b-\widehat{b})}+\frac{1}{Q(Q-1)}\xi_Q^{(u-\widehat{u})\sigma(i^{\prime})}\sum_{j=0}^{Q-1}1\\
&=&\frac{1}{Q-1}\xi_Q^{(u-\widehat{u})\sigma(i^{\prime})}.
\end{eqnarray*}

When $i^{\prime}=\ell$,  we have
\begin{eqnarray*}
\cF_1\cF_2^H&=&\frac{1}{Q^(Q-1)}\sum_{i=0,i\neq\ell}^{Q-1}\xi_Q^{(u-\widehat{u})\sigma(i)}\sum_{j=0}^{Q-1}\xi_Q^{j((a-\widehat{a})\pi(i)+b-\widehat{b})}\\
&=&0.
\end{eqnarray*}
Hence, we get that
$$\left|\cF_1\cF_2^H\right|\in\left\{0,\frac{1}{Q-1}\right\}.$$

Summarizing the conclusions in the three cases above, we
obtain
$$I_{\max}(\cc)=\frac{1}{Q-1}.$$
This completes the proof.
\end{proof}


\begin{thm} \label{thm4}
Let symbols be the same as before. If $Q>2$, then the codebook $\cc$ is asymptotically optimal with respect to the Welch bound.
\end{thm}
\begin{proof}
For $N=p_{{\min}}Q^2+Q^2-Q$ and $K=Q(Q-1)$, by Lemma \ref{welch} we have
$$
I_W=\sqrt{\frac{p_{{\min}}Q}{(p_{{\min}}Q^2+Q^2-Q-1)(Q-1)}}.
$$
Hence,
\begin{eqnarray*}
\frac{I_{\max}(\cc)}{I_W}&=&\sqrt{\frac{p_{{\min}}Q^2+Q^2-Q-1}{Q(Q-1)p_{{\min}}}}\\
&=&\sqrt{\frac{Q}{Q-1}+\frac{Q}{(Q-1)p_{{\min}}}-\frac{1}{(Q-1)p_{{\min}}}-\frac{1}{Q(Q-1)p_{{\min}}}}.
\end{eqnarray*}
Consequently,
$$\lim_{p_{{\min}}\rightarrow+\infty }\frac{I_{\max}(\cc)}{I_W}=1.$$
This completes the proof of this theorem.
\end{proof}
Table \ref{tab3} presents some   parameters of  codebooks derived from Theorem  \ref{thm3}. From Table \ref{tab3},  we can see  $I_W$  is very close to $I_{\max}(\cc)$  as $p_{\min}$ increases. This  means that the codebooks defined in Theorem \ref{thm3}  are asymptotically optimal with respect to the Welch bound for large $p_{\min}$, as predicted in Theorem \ref{thm4}.
\begin{table}[!htbp]
 \caption{\small The parameters of the $(N,K)$ codebook in Theorem \ref{thm3}}
 \label{tab3}
 \centering
 \begin{tabular}{|l|l|l|l|l|l|l|}
   \hline
   $p_{\min}$& $Q$   &$N$ &$K$ &$I_{\max}$ & $I_W$& $I_W/I_{\max}$ \\ \hline
   7& $77$   & $47355$  &5852 & 0.013072& 0.01224 & 0.93618  \\ \hline
   19 & $437$ &  $3818943$ & 190532&$0.22906\times 10^{-2}$ & $0.22331\times 10^{-2}$&  0.97474  \\ \hline
   29 & $1073$ &  $34538797$  & 1150256& $0.93240\times 10^{-3}$&$0.91674\times 10^{-3} $& 0.98321 \\ \hline
    41 & $2173$ &  $198318845$  &4719756 &$0.46023\times 10^{-3}$ & $0.45479\times 10^{-3}$& 0.98803 \\ \hline
   59 & $3599$  & $777164461$  & 12949202& $0.27789\times 10^{-3}$&$0.27557\times 10^{-3}$ & 0.99163 \\ \hline
  67 & $4757$   & 1538770575 & 22624292&$0.21024\times 10^{-3}$ &$0.20869\times 10^{-2}$ & 0.99262 \\ \hline
   79 & $6557$  & 3439533363 & 42987692&$0.15252\times 10^{-3}$ &$0.15156\times 10^{-3}$ & 0.99373 \\ \hline
    89 & $8633$  & 6707573377 &74520056 &$0.11584\times 10^{-3}$ &$0.11520\times 10^{-3}$ &0.99443\\ \hline
   101 & $11009$ &12362193253 & 121187072&$0.90839\times 10^{-4}$ &$0.90393\times 10^{-4}$ &0.99509 \\ \hline
 \end{tabular}
 \end{table}

\section{Concluding remarks}
Employing generalised bent functions, we presented  two classes of codebooks and proved that these constructed codebooks nearly meet the Welch bound.
As a comparison, the parameters of some known classes of asymptotically optimal codebooks with respect to the Welch bound and the new ones are listed in Table \ref{CCCtabl1}.
Obviously, the parameters of our classes of codebooks have not been covered in the literature.


\end{document}